\documentclass[a4paper,UKenglish,cleveref, autoref, thm-restate]{lipics-v2021}

\pdfoutput=1 
\hideLIPIcs  

\usepackage{mathtools}

\DeclareMathOperator*{\poly}{poly}
\DeclareMathOperator*{\polylog}{polylog}

\usepackage[ruled,vlined,linesnumbered,noend]{algorithm2e}
\SetKwInOut{Parameter}{Parameters}

\def\EMPH#1{\emph{#1}}

\newcommand{\de}{\mathrm{d}}
\newcommand{\ex}{\mathbb{E}}
\newcommand{\pr}{\mathbb{P}}

\newcommand{\eun}{E_\mathrm{un}}

\newcommand{\ecl}{E_\mathrm{cl}}

\newcommand\myeq{\stackrel{\mathclap{\normalfont\mbox{(Lemma \ref{IMP Inequaly})}}}{=}}

\newcommand\myobeq{\stackrel{\mathclap{\normalfont\mbox{(Observation \ref{Path Bound})}}}{=}}

\newcommand\myineq{\stackrel{\mathclap{\normalfont\mbox{(Lemma \ref{IMP Inequaly})}}}{\leq}}

\newcommand\mydineq{\stackrel{\mathclap{\normalfont\mbox{(Claim \ref{degrees})}}}{\leq}}

\title{Density-Sensitive Algorithms for $(\Delta + 1)$-Edge Coloring} 

\titlerunning{Density-Sensitive Algorithms for $(\Delta + 1)$-Edge Coloring} 

\author{Sayan Bhattacharya}{University of Warwick, United Kingdom}{s.bhattacharya@warwick.ac.uk}{}{}

\author{Mart\'{i}n Costa}{University of Warwick, United Kingdom}{martin.costa@warwick.ac.uk}{}{}

\author{Nadav Panski}{Tel Aviv University, Israel}{nadavpanski@mail.tau.ac.il}{}{}

\author{Shay Solomon}{Tel Aviv University, Israel}{shayso@tauex.tau.ac.il}{}{}

\authorrunning{S. Bhattacharya, M. Costa, N. Panski and S. Solomon} 

\Copyright{Sayan Bhattacharya, Mart\'{i}n Costa, Nadav Panski, and Shay Solomon} 



\category{} 

\relatedversion{} 

\funding{Shay Solomon is funded by the European Union (ERC, DynOpt, 101043159).
Views and opinions expressed are however those of the author(s) only and do not necessarily reflect those of the European Union or the European Research Council. Neither the European Union nor the granting authority can be held responsible for them.
Shay Solomon and Nadav Panski are supported by the Israel Science Foundation (ISF) grant No.1991/1.
Shay Solomon is also supported by a grant from the United States-Israel Binational Science Foundation (BSF), Jerusalem, Israel, and the United States National Science Foundation (NSF).}

\nolinenumbers 

\EventEditors{John Q. Open and Joan R. Access}
\EventNoEds{2}
\EventLongTitle{42nd Conference on Very Important Topics (CVIT 2016)}
\EventShortTitle{ESA 2024}
\EventAcronym{ESA}
\EventYear{2024}
\EventDate{December 24--27, 2016}
\EventLocation{Little Whinging, United Kingdom}
\EventLogo{}
\SeriesVolume{42}
\ArticleNo{23}

\begin{document}

\maketitle

\begin{abstract}
Vizing's theorem asserts the existence of a \EMPH{$(\Delta+1)$-edge coloring} for any graph $G$, where $\Delta = \Delta(G)$ denotes the maximum degree of $G$. Several polynomial time $(\Delta+1)$-edge coloring algorithms are known, and the state-of-the-art running time (up to polylogarithmic factors) is $\tilde{O}(\min\{m \sqrt{n}, m \Delta\})$,\footnote{Here and throughout the $\tilde{O}$ notation suppresses $\polylog(n)$ factors.} by Gabow, Nishizeki, Kariv, Leven and Terada from 1985, where $n$ and $m$ denote the number of vertices and edges in the graph, respectively. Recently, Sinnamon shaved off a $\polylog(n)$ factor from the time bound of Gabow et al.

The {\em arboricity} $\alpha = \alpha(G)$ of a graph $G$ is the minimum number of edge-disjoint forests into which its edge set can be partitioned, and it is a measure of the graph's ``uniform density''. While $\alpha \le \Delta$ in any graph, many natural and real-world graphs exhibit a significant separation between $\alpha$ and $\Delta$.

In this work we design a $(\Delta+1)$-edge coloring algorithm with a running time of $\tilde{O}(\min\{m \sqrt{n}, m \Delta\})\cdot \frac{\alpha}{\Delta}$, thus improving the longstanding time barrier by a factor of $\frac{\alpha}{\Delta}$. In particular, we achieve a near-linear runtime for bounded arboricity graphs (i.e., $\alpha = \tilde{O}(1)$) as well as when $\alpha = \tilde{O}(\frac{\Delta}{\sqrt{n}})$. Our algorithm builds on Gabow et al.'s and Sinnamon's algorithms, and can be viewed as a density-sensitive refinement of them.
\end{abstract}

\section{Introduction}
A \EMPH{(proper) $k$(-edge) coloring} in a graph $G$ is a coloring of edges,
where each edge is assigned a color from 
the set $[k] := \{1,\ldots,k\}$,
such that no two adjacent edges have the same color.
Clearly, $k$ must be at least as large as the maximum degree $\Delta = \Delta(G)$ of the graph $G$, and Vizing's theorem \cite{vizing1964estimate} states that $\Delta+1$ colors always suffice; in some graphs $\Delta+1$ colors are necessary.

Several polynomial time $(\Delta+1)$-coloring algorithms are known, 
including a simple $O(m n)$-time algorithm by Misra and Gries \cite{MG92}, which is a simplification of an earlier algorithm by Bollob{\'a}s \cite{Bol82}.
In 1985  Gabow, Nishizeki, Kariv, Leven and Terada \cite{Gabow85} presented a $(\Delta+1)$-coloring
algorithm with a running time of $O(\min\{m \sqrt{n \log n}, m \Delta \log n\})$.
A recent work by Sinnamon \cite{Sin} shaves off some $\polylog(n)$ factors; specifically, Sinnamon removed the $\sqrt{\log n}$ factor from the term $m \sqrt{n \log n}$, to achieve a clean runtime bound of $O(m \sqrt{n})$.
Nonetheless, up to the $\polylog(n)$ factors, no improvement on the runtime of the algorithm of \cite{Gabow85} was reported to date.
We summarize this state-of-the-art result in the following theorem:
\begin{theorem} [Gabow et al.\ \cite{Gabow85}] \label{gab}
For any $n$-vertex $m$-edge graph of maximum degree $\Delta$,
a $(\Delta+1)$-edge coloring can be computed within time $\tilde{O}(\min\{m \sqrt{n}, m \Delta\})$.

\end{theorem}

Note that the runtime bound provided by Theorem \ref{gab}  
is near-linear in bounded degree graphs.
However, in most graphs of interest, the maximum degree $\Delta$ is large.
The question of whether or not one can significantly improve this runtime bound in graphs of large maximum degree $\Delta$ has remained open.

\medskip
\noindent \textbf{Bounded Arboricity.}
Sparse graphs, or graphs of ``low density'', are of importance in both theory and practice. A key  definition that   captures the property of low density in a ``uniform manner'' is bounded \EMPH{arboricity}, which constrain the average degree of any subgraph.

\begin{definition} Graph $G$ has \EMPH{arboricity} $\alpha = \alpha(G)$ if
 $\frac{m_s}{n_s-1}\leq \alpha$, for every $S\subseteq V$, where $m_s$ and $n_s$ are the number of edges and vertices in the graph induced by $S$, respectively.
Equivalently (by the Nash-Williams theorem \cite{nash1964}), the edges of a graph of arboricity $\alpha$ can be decomposed into $\alpha$ edge-disjoint forests.
\end{definition}

While $\alpha \le \Delta$ holds
in any graph $G$,
there might be a large separation between $\alpha$ and $\Delta$; e.g., for the $n$-star graph we have $\alpha = 1, \Delta = n-1$. 
A large separation between $\alpha$ and $\Delta$ is exhibited 
in  many natural and
real-world graphs, such as the world wide web graph, social networks and transaction networks, as well as in various random distribution models, such as the preferential attachment model. 
Note also that the family of bounded arboricity graphs, even for $\alpha = O(1)$, includes all graphs that exclude a fixed minor, which, in turn include all bounded treewidth and bounded genus graphs, and in particular all planar graphs. 

In this work we present a near-linear time 
$(\Delta+1)$-coloring algorithm in graphs of arboricity $\alpha = \tilde{O}(1)$.
Further, building on our new algorithm,  
we present an algorithm that improves over the longstanding time barrier (provided by Theorem \ref{gab}) by a factor of $\frac{\alpha}{\Delta}$, as summarized in the following theorem.
\begin{theorem} \label{th:main}
For any $n$-vertex $m$-edge graph of maximum degree $\Delta$ and arboricity $\alpha$,
there is a randomized algorithm that computes a $(\Delta+1)$-edge coloring within time $\tilde{O}(\min\{m \sqrt{n} \cdot \frac{\alpha}{\Delta},
m \alpha\})
= \tilde{O}(\min\{m  \sqrt{n}, m \Delta\}) \cdot \frac{\alpha}{\Delta}$. The time bound holds both in expectation and with high probability.
\end{theorem}
\noindent
{\bf Remark.} The exact runtime bound of our algorithm (including the $\polylog(n)$ factors that are suppressed under the $\tilde{O}$-notation in Theorem \ref{th:main})
    is the following:
    $O(\min\{m  \sqrt{n \log n}, m \Delta \log n\}) \cdot \frac{\alpha}{\Delta} + O(m  \log n)$ in expectation and $O(\min\{m  \sqrt{n} \log^{1.5} n, m \Delta \log^2 n\}) \cdot \frac{\alpha}{\Delta} + O(m \log^2 n)$ with high probability. However, we made no 
    attempt to optimize polylogarithmic factors in this work. Note that this runtime is near-linear when $\alpha = \tilde{O}(1)$ as well as when $\alpha = \tilde{O}(\frac{\Delta}{\sqrt{n}})$.
    

\medskip
\noindent
The aforementioned improvement of Sinnamon \cite{Sin} to the
state-of-the-art runtime bound by Gabow et al.\ \cite{Gabow85} 
(i.e., the removal of the $\sqrt{\log n}$ factor from the term $m \sqrt{n \log n}$)
was achieved via two algorithms: 
a rather intricate deterministic algorithm, which is very similar to a deterministic algorithm by \cite{Gabow85}, 
and an elegant randomized algorithm 
that greatly simplifies the deterministic one.
Our algorithm follows closely Sinnamon's randomized algorithm, with one key difference: we {\em give precedence to low degree vertices and edges over high degree ones}, where the degree of an edge (which we shall refer to as its \EMPH{weight}) is the minimum degree of its endpoints.
Our algorithm can thus be viewed as a \EMPH{degree-sensitive refinement} of Sinnamon's algorithm.
The analysis of our algorithm 
 combines several new ideas  to achieve the claimed improvement in the running time.

\subsection{Technical Overview and Conceptual Contribution}

To compute a $(\Delta+1)$(-edge)-coloring, one can simply color the edges of the graph one after another, using what we shall refer to as a  \texttt{Color-One-Edge} procedure.
The most basic \texttt{Color-One-Edge} procedure is by Misra and Gries \cite{MG92}, and is a simplification of an algorithm by Bollob{\'a}s \cite{Bol82}. 
Given a graph with some \emph{partial} $(\Delta+1)$-coloring and an arbitrary uncolored edge $e = (u,v)$, this procedure {\em recolors} some edges so as to {\em free} a color for the uncolored edge $e$, and then colors $e$ with that color.
Procedure \texttt{Color-One-Edge} is carried out by (1) creating a \EMPH{fan} $F$ centered at one of $e$'s endpoints $u$ and {\em primed} by some color $c_{1}$,\footnote{Such a fan is a star rooted at $u$ that spans some of $u$'s neighbors, with specific conditions on the colors of the edges of this star and the missing colors on their vertices.
Refer to Section \ref{sec:prel} for the definition of this and all other notions used throughout.)} and (2) a simple \EMPH{maximal alternating path} $P$ starting at $u$ with edges colored by the primed color $c_{1}$ and another color $c_{0}$ that is {\em missing} (i.e., not occupied) on $u$. 
The runtime of Procedure \texttt{Color-One-Edge} is linear in the size of $F$ and the length of $P$, 
hence it is $O(|F|+|P|)=O(\de(u)+n)=O(n)$, where $\de(u)$ denotes the degree of $u$.
Applying this procedure iteratively for all edges in the graph leads to a runtime of $O(m n)$.

Instead of coloring one   edge at a time
via Procedure \texttt{Color-One-Edge},
Gabow et al.\ \cite{Gabow85} proposed a different approach, which uses a more complex procedure, \texttt{Parallel-Color}, for coloring {\em multiple} uncolored edges at a time. 
Procedure \texttt{Parallel-Color} 
chooses a color $c$ and colors, in $O(m)$ time, a constant fraction of the uncolored edges incident to vertices on which color $c$ is missing.
By applying Procedure \texttt{Parallel-Color} iteratively $O(\Delta\log n)$ times, all edges can be colored in $O(m \Delta \log n)$ time. 
This $O(m\Delta\log n)$-time algorithm, hereafter   \texttt{Low-Degree-Color},
is fast only for graphs of small maximum degree. 

Gabow et al.\ \cite{Gabow85} also gave a simple recursive algorithm, hereafter \texttt{Recursive-Color-Edges}, which first partitions the edges of the graph into two separate subgraphs of maximum degree $\le \Delta/2 + 1$, 
then recursively computes a $(\Delta/2 +2)$-coloring in each subgraph via \texttt{Recursive-Color-Edges}, and it combines these two colorings into a single $(\Delta+4)$-coloring of the entire graph. Next, the algorithm uncolors all edges in the three smallest color class (of size $O(\frac{m}{\Delta})$),
and finally all uncolored edges are then colored 
in time $O(\frac{mn}{\Delta})$
by applying the aforementioned  \texttt{Color-One-Edge}. The recursion bottoms when the maximum degree is small enough ($\Delta \le \sqrt{\frac{n}{\log n}}$), and then Algorithm  \texttt{Low-Degree-Color} is applied in order to color the remaining subgraphs.
As the term of $O(\frac{mn}{\Delta})$ grows geometrically with the recursion levels and as
the recursion bottoms at maximum degree  $O(\sqrt{\frac{n}{\log n}})$,
the total runtime of \texttt{Recursive-Color-Edges}
is $O(m\sqrt{n\log n})$.

Sinnamon \cite{Sin} 
obtained a runtime of  $O(m\sqrt{n})$ time.
He achieved this result via two algorithms: A deterministic algorithm, which is similar to that of Gabow et al., and a 
much simpler and elegant randomized variant, which we briefly describe next.
First, Sinnamon devised a simple random version of Procedure \texttt{Color-One-Edge}, where the only difference from the deterministic procedure is that the uncolored edge $e = (u,v)$, the endpoint $u$ of $e$, and the missing color of $u$ --- are all chosen {\em randomly} rather than {\em arbitrarily}.
Sinnamon observes that the runtime of this random procedure is not just $O(n)$ as before, but it is also 
bounded by $O(\frac{m\Delta}{l})$ {\em in expectation}, where $l$ is the number of uncolored edges in the graph.
Indeed, recall that the runtime of this procedure is linear in the size of the fan $F$ and the length of the path $P$. Clearly, the size of a fan is bounded by the maximum degree $\Delta = O(\frac{m\Delta}{l})$. The main observation is that one can bound the expected length of the path $P$ 
by $O(\frac{m\Delta}{l})$
as follows: (i) as each edge can be in at most $\Delta+1$ such paths, the sum of lengths of all the possible maximal alternating paths is bounded by $m(\Delta+1)=O(m\Delta)$, and (ii) the probability of the algorithm choosing any particular path is bounded by $\frac{1}{l}$.

By iteratively applying this randomized procedure \texttt{Color-One-Edge}, the expected runtime of coloring all edges 
is bounded by $\sum_{i=1}^{m}O(\frac{m\Delta}{m+1-i})=O(m\Delta\sum_{i=1}^{m}\frac{1}{m+1-i})=O(m\Delta\log n)$. 
This provides a much simpler randomized substitute for 
Algorithm \texttt{Low-Degree-Color} by \cite{Gabow85}, which Sinnamon then applied in conjunction with  
Algorithm \texttt{Recursive-Color-Edges} by \cite{Gabow85} and some small tweaks to obtain an elegant randomized algorithm with expected runtime $O(m \sqrt{n})$.


\subsubsection{Our approach}
As mentioned, the goal of this work
is to improve the longstanding runtime bound for $(\Delta+1)$-coloring \cite{Gabow85,Sin} by a factor of $\frac{\alpha}{\Delta}$, where $\alpha$ is the graph's arboricity. 
While the required number of colors ($\Delta$ or $\Delta+1$) grows linearly with the maximum degree $\Delta$, 
the runtime might not need to grow too. 
The main conceptual contribution of this work is in unveiling this rather surprising phenomenon --- \EMPH{the running time shrinks as the maximum degree grows} (provided that the arboricity does not grow together with $\Delta$).

In low arboricity graphs, although some vertices may have high degree (close to $\Delta$), most vertices have low degree (close to $\alpha$). Similarly, defining the degree (or {\em weight}) $w(e)$ of an edge $e$ as the minimum degree of its endpoints --- although some edges may have large weight, most edges have low weight (as the sum of edge weights is known to be $O(m \alpha)$). 
Our key insight is that
giving precedence to low degree vertices and edges over high degree ones gives rise to significant improvements in the running time, for all graphs with $\alpha \ll \Delta$. 


\medskip
\noindent \textbf{A density-sensitive \texttt{Low-Degree-Color} algorithm.}
Recall that the only difference between Sinnamon's 
\texttt{Color-One-Edge} procedure
and the deterministic one by \cite{Gabow85} is that the uncolored edge $e= (u,v)$, the endpoint $u$ of $e$, and the missing color of $u$ are all chosen {\em randomly} rather than {\em arbitrarily}. 
We further modify the randomized procedure of \cite{Sin} by actually \EMPH{making one of these random choices deterministic} --- 
we choose $u$ as the endpoint of $e$ of minimum degree, so that $\de(u) = w(e)$. In this way the expected size of the fan computed by the algorithm is bounded by the expected weight of the sampled edge, which, in turn, is the total weight of all uncolored edges divided by their number, yielding an upper bound of $O(\frac{m\alpha}{l})$ --- which refines the aforementioned bound of $O(\frac{m\Delta}{l})$ by \cite{Sin}.
Next, we would like to argue that the expected length of a maximal alternating path $P$ is $O(\frac{m\alpha}{l})$. While the argument for bounding the probability that the algorithm chooses any particular path is 
$O(\frac{1}{l})$ remains pretty much the same, the challenge is in bounding the
sum of lengths of all maximal alternating paths by $O(m\alpha)$ rather than $O(m \Delta)$.
Our key insight here is 
to bound the total lengths of the \EMPH{internal} parts of the paths, by exploiting the observation that \EMPH{any edge $e$ can be internal in at most $w(e)$ maximal alternating paths}, which directly implies that
the sum of lengths is bounded by the sum of the weights of the colored edges, which, in turn is bounded by $O(m\alpha)$.
As a result, we improve the runtime of Procedure \texttt{Low-Degree-Color} from $O(m\Delta\log n)$ to $O(m\alpha\log n)=O(m\Delta\log n)\cdot\frac{\alpha}{\Delta}$. Refer to Section \ref{sec:basicalg} for the full details.

\medskip
\noindent \textbf{A density-sensitive \texttt{Recursive-Color-Edges} algorithm.}
The aforementioned improvement provided by our density-sensitive
\texttt{Low-Degree-Color} algorithm
only concerns the regime of low $\Delta$.
To shave off a factor of $\frac{\alpha}{\Delta}$ in the entire regime of $\Delta$, we need to come up with a density-sensitive  \texttt{Recursive-Color-Edges} algorithm --- which is the main technical challenge of this work.

In the original recursive algorithm of Gabow et al.\ \cite{Gabow85}, the starting point is that any graph can be partitioned in linear time into two subgraphs in which the degrees of all vertices, including the maximum degree, reduce by a factor of 2.
To achieve a density-sensitive refinement of Algorithm \texttt{Recursive-Color-Edges}, the analog starting point would be to partition the graph into two subgraphs of half the arboricity. Although this guarantee can be achieved via a random partition, we overcome this issue deterministically, by \EMPH{transitioning from the graph's arboricity to the graph's normalized weight}, where the {\em normalized weight} of the graph is defined as the ratio of the graph's weight (the sum of its edge weights) to the number of edges.
Indeed, since the degrees of vertices decay by a factor of 2 at each recursion level, so does the weight of each subgraph.
An added benefit of this approach is that all our results extend from graphs of arboricity $\alpha$ to graphs of normalized weight $2\alpha$, which forms a much wider graph class.\footnote{Any graph of arboricity $\alpha$ has normalized weight at most $2\alpha$ (refer to Claim \ref{cl:basearb}), but a graph of normalizd weight $\alpha$ may have arboricity much larger than $\alpha$.}

In the original recursive algorithm of \cite{Gabow85}, to transition from a $(\Delta+4)$-coloring to a $(\Delta+1)$-coloring, the algorithm uncolors all edges in the three color classes of smallest size,  and then colors all those $O(\frac{m}{\Delta})$ edges by applying the \texttt{Color-One-Edge} procedure, which takes time $O(\frac{m n}{\Delta})$.
To outperform the original algorithm, our algorithm uncolors all edges in the three color classes of smallest \EMPH{weight} (rather than size), implying that the total weight of those uncolored edges is bounded by $O(\frac{m\alpha}{\Delta})$.
As a result, we can no longer argue that the total number of those edges is $O(\frac{m}{\Delta})$ as before, but we demonstrate that the new upper bound on the weights gives rise to a faster algorithm. 

We improve the bound for coloring the uncolored edges at every level of recursion, from $O(\frac{mn}{\Delta})$ to $O(\frac{mn}{\Delta}) \cdot \frac{\alpha}{\Delta}$, as follows.
First, consider the time spent on uncolored edges of \EMPH{high weights}, i.e., edges of weight $> \Delta/2$. Since the number of such edges is bounded by $O(\frac{m \alpha}{\Delta}) / (\Delta/2) = O(\frac{m}{\Delta})\cdot \frac{\alpha}{\Delta}$, the simple $O(n)$ time bound of \texttt{Color-One-Edge} by \cite{MG92} yields the required total runtime bound of 
$O(\frac{mn}{\Delta}) \cdot \frac{\alpha}{\Delta}$.

The challenging part is to bound the time spent on uncolored edges of \EMPH{low weights}, i.e., edges of weight  $\le \Delta/2$.
Towards this end, we employ a stronger version of the bound on the expected runtime of \texttt{Color-One-Edge}. Instead of using the $O(\frac{m\alpha}{l})$ bound mentioned before, we build on the fact that any chosen endpoint of a low weight edge has $\Omega(\Delta)$ possible missing colors to choose from, which reduces the probability to choose any alternating path by a factor of $\Delta$. With a bit more work, we can show that the expected time for coloring a low weight uncolored edge is bounded by $O(\frac{m\alpha}{l\cdot\Delta})$, yielding a total runtime bound of at most $O(\sum_{i=1}^{m}\frac{m\alpha}{(m+1-i)\cdot\Delta})
=O(\frac{m\alpha}{\Delta}\log n)$ for all the low weight uncolored edges.
Although this upper bound is better than the required one by factor of $\frac{n}{\Delta}$, the bottleneck stems from the edges of high weights.

As we demonstrate in Section \ref{sec:advanced}, by combining these ideas with a careful analysis, our randomized \texttt{Recursive-Color-Edges}
Algorithm achieves the required runtime bound of $O(m\sqrt{n\log n})\cdot\frac{\alpha}{\Delta}$. 

\subsection{Related Work}
Edge coloring is a fundamental graph problem that has been studied over the years
in various settings and computational models;
we shall aim for brevity.
Surprisingly perhaps, the body of work on fast sequential $(\Delta+1)$-coloring algorithms, which is the focus of this work, is quite sparse. 
We have already discussed in detail the $(\Delta+1)$-coloring algorithms of \cite{Gabow85,MG92,Sin}; a similar result to Gabow et al.\ \cite{Gabow85} was obtained independently (and earlier) by \cite{arjomandi1982efficient}.
It is NP-hard to determine whether a graph can be colored with $\Delta$ colors or not \cite{holyer1981np}.
On the other hand, for restricted graph families, particularly biparite graphs and planar graphs of sufficiently large degree, 
near-linear time algorithms are known
\cite{gabow1976using,cole2008new},\cite{chrobak1990improved}.
Recently, $(\Delta+1)$-coloring algorithms across different models of computation were given in \cite{Christiansen22}.

There are also known edge coloring algorithms for bounded arboricity graphs, but they are slow or they use many colors.
In graphs of {\em constant} arboricity $\alpha$,  
a near-linear time 
$(\lceil \frac{(\alpha+2)^2}{2}\rceil - 1)$-coloring
algorithm 
 was given in \cite{zhou1994edge}.
For graphs whose arboricity $\alpha$ is smaller than  $\Delta$ by at least a constant factor,
several $\Delta$-coloring algorithms
are known \cite{fiorini1975some,haile1999bounds,sanders2002size,woodall2007average,woodall2008average,cao2019average},
 but their running time is at least $O(m n)$.

\subsection{Subsequent Work}
Kowalik \cite{kowalik2024edge} gave a $\Delta$-edge coloring algorithm with a running time of $\tilde O(m \alpha^3)$,
for graphs whose arboricity $\alpha$ is smaller than  $\Delta$ by at least a constant factor.\footnote{More precisely, the result holds for graphs whose {\em maximum average degree} is at most $\Delta/2$; the maximum average degree is within a constant factor from the graph's arboricity.}
The work of \cite{kowalik2024edge} crucially builds on top of our techniques; indeed, without relying on our techniques, the running time of the algorithm of \cite{kowalik2024edge} would grow to $\tilde{O}(m \Delta^3)$, which is $\tilde{O}(m  n^3)$ in bounded arboricity graphs of large degree; this bound, in turn, is slower by a factor of $n^2$ than some of the aforementioned $\Delta$-edge coloring algorithms in bounded arboricity graphs \cite{fiorini1975some,haile1999bounds,sanders2002size,woodall2007average,woodall2008average,cao2019average}. 

For a dynamically changing graph with maximum degree $\Delta$
and arboricity $\alpha$,
algorithms for maintaining a $(\Delta +O(\alpha))$-coloring
with update time $\tilde{O}(1)$ were given independently in \cite{CRV23} and \cite{BCPS23}.

Very recently, Bhattacharya et al.~\cite{BhattacharyaCCSZ24} and Assadi~\cite{Assadi24} independently obtained faster algorithms for $(\Delta + 1)$-edge coloring in general graphs, obtaining running times of $\tilde O(mn^{1/3})$ and $\tilde O(n^2)$ respectively.

\section{Preliminaries} \label{sec:prel}

We work in the standard word RAM model of computation, with words of size $w := \Theta(\log n)$. In particular, we can index any of the $2^{O(w)} = \poly(n)$ memory addresses, perform basic
arithmetic on $w$-bit words, and sample $w$-bit random variables, all in constant time.
In App.\ \ref{Data Sets}, we describe the data structures used by our algorithms; they are all basic and easy to implement, and their space usage is linear in the graph size.

We denote the degree of a vertex $v$ by $\de(v)$.
The \EMPH{weight $w(e)$} of an edge $e = (u,v)$ is defined as the minimum degree of its   endpoints, i.e., 
    $w(e):=\min\{\de(u),\de(v)\}$.
    The \EMPH{weight $w(G)$} of a graph $G = (V,E)$ is defined as the sum of weights over its edges, i.e., 
    $w(G) :=\sum_{e \in E} w(e)$.

Note that the weight $w(G)$ of any $m$-edge graph $G$ satisfies $w(G) \ge m$.
The following claim, due to \cite{ChibaN85}, asserts that the weight of any $m$-edge graph exceeds $m$ by at most a factor of $2\alpha$.
\begin{claim} [Lemma 2 in 
\cite{ChibaN85}] \label{cl:basearb}
    For any $m$-edge graph $G$ with arboricity $\alpha$, we have  $w(G)\le 2m\alpha$.
\end{claim}

In what follows, let $G=(V,E)$ be an arbitrary $n$-vertex $m$-edge graph, and we let $\Delta$ to the maximum degree of $G$.
For any integer $k \ge 1$, let $[k]$ denote the set $\{1,2,...,k\}$.

\begin{definition}
    A \EMPH{(proper) partial k-(edge-)coloring} $\chi$ of $G$ is a color function $\chi:E\rightarrow [k]\cup \{\neg\}$ such that any two distinct colored edges $e_1, e_2$ that share an endpoint do not receive the same color. An edge $e$ with $\chi(e)\in [k]$ is said to be \EMPH{colored} (by $\chi$), whereas an edge $e$ with $\chi(e)=\neg$ is said to be \EMPH{uncolored}.
    If all the edges in $G$ are colored by $\chi$, we say that $\chi$ is a \EMPH{(proper) $k$-(edge-)coloring}.
\end{definition}
 
    Given a partial $k$-coloring $\chi$ for $G$, we define
    \EMPH{$M(v)$} as the set of \EMPH{missing} colors of $v$,
    i.e., the set of colors in the color palette $[k]$ not occupied
    by any of the incident edges of $v$.
For a partial ($\Delta+1$)-coloring, $M(v)$ is always nonempty, as $v$ has at most $\Delta$ neighbours and there are $\Delta + 1$ colors to choose from.

\begin{definition} [Colored and uncolored edges]
    Consider a partial $k$-coloring $\chi$, where all the edges in $G$ are colored but $l$. We define \EMPH{$\ecl(G,\chi) =  \ecl$} as the set of $m - l$ colored edges of $G$ and \EMPH{$\eun(G,\chi) = \eun$} as the set of $l$ uncolored edges of $G$. If all edges of $G$ are uncolored, i.e., $l=m$,
    we say that $\chi$ is an \EMPH{empty} coloring.
\end{definition}

\subsection{Fans}
\label{Fans}

In what follows we let $\chi$ be a proper partial ($\Delta+1$)-coloring of $G$.

\begin{definition}
    A \EMPH{fan} $F$ is a sequence of vertices $(v,x_{0},...,x_{t})$ such that $x_{0},...,x_{t}$ are distinct neighbors of v, the edge $(v,x_{0})$ is uncolored, the edge $(v,x_{i})$ is colored for every $i \in [t]$, and the color $\chi(v,x_{i})$ is missing at vertex $x_{i-1}$ for every $i \in [t]$.
    The vertex $v$ is called the {\em center} of $F$, and $x_{0},...,x_{t}$ are called the {\em leaves} of $F$.
\end{definition}

Refer to Figure \ref{fig:fan} in App.\ \ref{app:figs} for an illustration of a fan.
The useful property of a fan is that \EMPH{rotating} or \EMPH{shifting} the colors of the fan preserves the validity of the coloring.
Let $F=(v,x_{0},...,x_{t})$ be a fan and write $c_{i} =\chi(v,x_{i})$, for each $i=1,...,t$.
To \EMPH{rotate} or \EMPH{shift} $F$ from $x_{j}$ means to set $\chi(v,x_{i-1})=c_{i}$ for every $i=1,...,j$ and make $(v,x_{j})$ uncolored.
After the rotation or shift, the function $\chi$ is still a proper partial coloring, but now $(v,x_{j})$ is uncolored instead of $(v,x_{0})$.
Note that $M(v)$ is unaffected by the shift.
Refer to Figure \ref{fig:fanShift} in App.\ \ref{app:figs} for an illustration of fan shifting.

To extend a partial ($\Delta+1$)-coloring
(i.e., increase the number of colored edges),
we shall focus on {\em maximal} fans (which cannot be further extended), using the following definition.

\begin{definition}
    A fan $F=(v,x_{0},...,x_{t})$ is said to be \EMPH{primed} by color $c_{1}\in M(x_{t})$ if one of the following two conditions hold:
    (1) $c_{1}\in M(v)$, or
    (2) $c_{1}\in M(x_{j})$ for some $j<t$.
\end{definition}

\noindent \textbf{Computing a primed fan.}
Given an uncolored edge $(v,x_{0})$, we can get a primed fan $(v,x_{0},...,x_{t})$ via \Cref{alg:makefan} (refer to \cite{Gabow85,MG92,Sin}):

\begin{algorithm}[h]
    \SetAlgoLined
    \DontPrintSemicolon
    \KwIn{A graph $G$ with a partial ($\Delta+1$)-coloring $\chi$ and an uncolored edge $(v,x_{0})$}
    \KwOut{A fan $F=(v,x_{0},...,x_{t})$ and color $c_{1}\in M(x_{t})$ such that $F$ is primed by $c_{1}$}
    $F\leftarrow (v,x_{0})$\;
    $t\leftarrow 0$\;
    \While{$F$ is not primed}{
        Pick any $c_{1} \in M(x_{t})$\;
        \If{$c_{1} \in M(v)$}{
            \Return{$F,c_{1}$}
        }
        \Else {
            Find the edge $e=(v,x_{t+1})$ such that $\chi(e)=c_{1}$\;
            \If{$x_{t+1}\in {x_{1},...,x_{t}}$}{
                \Return{$F,c_{1}$}
            }
            \Else{
                Append $x_{t+1}$ to $F$\;
                $t\leftarrow t+1$\;
            }
        }
    }
    \caption{\texttt{Make-Primed-Fan}$(G,\chi,(v,x_{0}))$}
    \label{alg:makefan}
\end{algorithm}

\begin{lemma}
\label{making fan bound}
   Algorithm \texttt{Make-Primed-Fan} returns a primed fan with center $v$ in $O(\de(v))$ time.
\end{lemma}

\begin{proof}
By the description of the algorithm,
it is immediate
that \texttt{Make-Primed-Fan} returns a primed fan.
The number of iterations of the while loop is bounded by $\de(v)$, since every iteration in which the loop does not terminate adds a new neighbor of $v$ as a leaf of $F$.
To complete the proof, we argue that each iteration can be implemented in constant time. 
We can store the at most $\de(v)$ leaves that are added to $F$ in a hash table, so that line 9 can be implemented in constant time.
The remaining part of an iteration can be carried out in constant time in the obvious way using the data structures mentioned in App.\ \ref{Data Sets}.
\end{proof}

\subsection{Alternating Paths}

\begin{definition}
    We say that path $P$ is a \EMPH{$(c_{0},c_{1})$-alternating path}, for a pair $c_0, c_1$ of distinct colors, if $P$ consists of edges with colors $c_{0}$ and $c_{1}$ only.
    We say that a $(c_{0},c_{1})$-alternating path $P$ is  \EMPH{maximal} if $P=(v_{0},v_{1},\ldots,v_{|P|})$ and both $v_{0}$ and $v_{|P|}$ have only one edge colored by $c_{0}$ and $c_{1}$ (hence $P$ cannot be extended further). Although a maximal alternating path may form a cycle, we shall focus on simple paths.
\end{definition}

The useful property of a maximal alternating path is that \EMPH{flipping} the colors of the path edges preserves the validity of the coloring.
That is, let $P = e_{1} \circ \ldots \circ e_{|P|}$ be a maximal $(c_{0},c_{1})$-alternating path such that for every $i=1,\ldots,|P|$: if $i \equiv 1 \pmod{2}$ then $\chi(e_{i})=c_{0}$ and if $i \equiv 0 \pmod{2}$ then $\chi(e_{i})=c_{1}$.
\EMPH{Flipping} $P$ means to set the   function $\chi$, such that for every $i=1,...,|P|$: if $i \equiv 1 \pmod{2}$ then $\chi(e_{i})=c_{1}$ and if $i \equiv 0 \pmod{2}$ then $\chi(e_{i})=c_{0}$;
the resulting function $\chi$ after the flip operation is a proper partial edge-coloring.

Flipping a maximal $(c_{0},c_{1})$-alternating path that starts at a vertex $v$ is a useful operation, as it "frees" for $v$ the color $c_0$ that is occupied by it, replacing it with the missing color $c_1$ on $v$.
Refer to Figure \ref{fig:pathFlip} in App.\ \ref{app:figs} for an illustration of maximal alternating path flipping.

\subsection{
Algorithm \texttt{Extend-Coloring}: Naively Extending a Coloring by One Edge}
 
Given an uncolored edge $e$, there is a standard way to color it and by that to extend the coloring, 
using a primed fan $F$ rooted at one of the endpoints of $e$, say $v$, as well as a maximal alternating path $P$ starting at $v$, by flipping the path and then rotating a suitable prefix of the fan, as described in the following algorithm (refer to 
\cite{Gabow85,MG92,Sin}):

\begin{algorithm}[t]
    \SetAlgoLined
    \DontPrintSemicolon
    \KwIn{A graph $G$ with a partial ($\Delta+1$)-coloring $\chi$, an uncolored edge e=$(v,x_{0})$, a fan $F=(v,x_{0},...,x_{t})$ primed by color $c_{1}$ and a maximal $(c_{0},c_{1})$-alternating path $P$ starting at $v$}
    \KwOut{Updated partial ($\Delta+1$)-coloring $\chi$ of $G$, such that e is also colored by $\chi$}
    \If{$c_{1}\in M(v)$}{
        Shift $F$ from $x_{t}$\;
        $\chi(v,x_{t})\leftarrow c_{1}$\;
    }
    \Else{
        $x_{j} \leftarrow$ leaf of $F$ with $\chi(v,x_{j})=c_{1}$ (when $j<t$)\;
        $w \leftarrow$ other endpoint of $P$ (other than $v$)\;
        Flip $P$\;
        \If{$w\ne x_{j-1}$}{
            Shift $F$ from $x_{j-1}$\;
            $\chi(v,x_{j-1})\leftarrow c_{1}$\;
        }
        \Else{
            Shift $F$ from $x_{t}$\;
            $\chi(v,x_{t})\leftarrow c_{1}$\;
        }
           
    }
    \Return{$\chi$}
    \caption{\texttt{Extend-Coloring}$(G,\chi,e,F,c_{1},P)$}
    \label{alg:extendcol}
\end{algorithm}

\begin{lemma}
    \label{Extend bound}
    Algorithm \texttt{Extend-Coloring} (\Cref{alg:extendcol}), when given as input a proper partial $(\Delta+1)$-coloring $\chi$, the first uncolored edge $(v,x_{0})$ of a fan $F$ primed by a color $c_1$ and a maximal $(c_0,c_1)$-alternating path $P$, where $c_0$ is free on $v$,
    extends the coloring $\chi$ into a proper partial $(\Delta+1)$-coloring, such that $e$, as well as all previously colored edges by $\chi$, is also colored by $\chi$.
    Moreover, this algorithm
    takes $O(|F|+|P|)=O(\de(v)+|P|)$ time. 
\end{lemma}

\begin{proof}
This algorithm provides the standard way to extend a coloring by one edge. (We omit the correctness proof for brevity; refer to
\cite{Gabow85,MG92,Sin} for the proof.)

As for the running time, there are three different cases to consider: (1) $c_{1}\in M(v)$, (2) $c_{1}\notin M(v)$ and $w\ne x_{j-1}$, 
(3) $c_{1}\notin M(v)$ and $ w= x_{j-1}$. Using our data structures (described in App.\ \ref{Data Sets}), one can identify the case and find the leaf $x_{j}$ and the endpoint $w$ (if needed) in time $O(|F|+|P|)$.
    Beyond that, in each case the algorithm performs at most one path flip, one fan shift, and one edge coloring, which also takes time $O(|F|+|P|)$.
    As $F$ contains at most $\de(u)$ leaves,
    the total running time is $O(|F|+|P|) = O(\de(v)+|P|)$.
\end{proof}

\section{A $(\Delta+1)$-Coloring Algorithm with Runtime $\tilde{O}(m \alpha)$} \label{sec:basicalg}

\subsection{Internal Edges of Maximal Alternating Paths} 

The paths that we shall consider are  maximal alternating paths that start and finish at different vertices (i.e., we do not consider cycles).
A vertex $v$ in a path $P$ is called \EMPH{internal} if it is not one of the two endpoints of $P$.
An edge $e = (u,v)$ of path $P$ is called \EMPH{internal} if both $u$ and $v$ are internal vertices of $P$.
For a path $P$, denote by 
    $I(P)$ the set of internal edges of $P$. 
We shall use the following immediate observation later. 
\begin{observation} \label{Path Bound}
    For any path $P$,
    $|P| \le |I(P)|+2$.
\end{observation}

Any edge may serve as an internal edge in possibly many different maximal alternating paths.
The following key lemma bounds the number of such paths by the edge's weight.
\begin{lemma}
\label{Edge Belonging}
    Any colored edge $e$ can be an internal edge of at most $w(e)$ maximal alternating paths.
\end{lemma}

\begin{proof}
Let $e=(u,v)$ be a colored edge with color $c_{e}$.
For any color $c_2 \ne c_e$, we note that $e$ can be internal in a maximal $(c_e,c_2)$-alternating path $P$ only if each among $u$ and $v$ is incident on an edge with color $c_{2}$,
thus the number of such colors $c_2$ is bounded by $\min\{\de(u),\de(v)\}=w(e)$.
To complete the proof, we note that there is at most one maximal $(c_e,c_2)$-alternating path that contains $e$, for any color $c_2$. 
\end{proof}

For a graph $G$ with a given coloring $\chi$, denote by $MP = MP(G,\chi)$ the set of maximal alternating paths in $G$ induced by $\chi$. 

\begin{lemma}
    \label{IMP Inequaly}
    $\sum_{P\in MP}|I(P)|\leq \sum_{e \in \ecl}w(e)$.
\end{lemma}

\begin{proof}
    \begin{eqnarray*}
        \sum_{P\in MP} |I(P)| ~=~ 
        \sum_{P \in MP} \sum_{e \in I(P)} 1 ~=~
        \sum_{e \in E} \sum_{P\in MP: e \in I(P)} 1
        ~=~
        \sum_{e \in \ecl} \sum_{P\in MP: e \in I(P)} 1
        ~\leq~
        \sum_{e \in \ecl} w(e),
    \end{eqnarray*}
    where the last inequality 
    follows from Lemma \ref{Edge Belonging},
    as $\sum_{P\in MP: e \in I(P)} 1$ is just the number of maximal alternating paths in which the colored edge $e$ is an internal edge.
\end{proof}

\subsection{Algorithm \texttt{Color-One-Edge}: An Algorithm for Coloring a Single Edge}
The following algorithm is based on Algorithm \texttt{Random-Color-One} of \cite{Sin}, but with one crucial tweak: Given the chosen random edge, we focus on the endpoint of the edge of {\em minimum degree}.

\begin{algorithm}[H]\label{alg:one}
    \SetAlgoLined
    \DontPrintSemicolon
    \KwIn{A graph $G$ with a partial ($\Delta+1$)-coloring $\chi$}
    \KwOut{Updated partial ($\Delta+1$)-coloring $\chi$ of $G$, such that one more edge is colored by $\chi$}
     $e=(u,v) \leftarrow$ A random uncolored edge\;
    Assume w.l.o.g.\ that $\de(u) \le \de(v)$\;
    $F,c_{1} \leftarrow$ \texttt{Make-Primed-Fan}$(G,\chi,(u,v))$\;
    Choose a random color $c_{0} \in M(u)$\;
    $P \leftarrow$ The maximal $(c_{0},c_{1})$-alternating path starting at $u$\;
    $\chi \leftarrow$\texttt{Extend-Coloring}$(G,\chi,e,F,c_{1},P)$\;
    \Return{$\chi$}
    \caption{\texttt{Color-One-Edge}$(G,\chi)$}
\end{algorithm}

\begin{lemma}
\label{Color-One-Edge Bound}
Let $G$ be a graph of maximum degree $\Delta$ and let $\chi$ be a partial $(\Delta + 1)$-coloring of all the edges in $G$ but $l$.
Then the expected runtime of \texttt{Color-One-Edge} on $G$ and $\chi$ is $O(\frac{w(G)}{l}) = O(\frac{m \alpha}{l})$.
\end{lemma}

\begin{proof}
Denote the runtime of Algorithm
\texttt{Color-One-Edge}
= \texttt{Color-One-Edge}$(G,\chi)$
by
$T($\texttt{Color-One-Edge}$)$.

Let $e_{r} = (u_{r},v_{r})$ be the random uncolored edge chosen in line 1 of the algorithm, with $\de(u_{r}) \le \de(v_{r})$, let $c_{r}$ be the random missing color chosen in line 4 of the algorithm, and let $P_{r}$ be the maximal alternating path starting at $u_{r}$ obtained in line 5 of the algorithm.
By definition,
both $|P_{r}|$ and $ \de(u_{r})$ are random variables.

    We will prove the lemma as a corollary of the following three claims:
    \begin{claim} \label{cl:total}
        $\ex[T(\texttt{Color-One-Edge})] = O(\ex[|P_{r}|]+ \ex[\de(u_{r})])$.
    \end{claim}
    \begin{claim} \label{cl:lengthP}
        $\ex[|P_{r}|]=O(1 + \frac{1}{l}\sum_{e \in \ecl} w(e))$.
    \end{claim}
    \begin{claim} \label{cl:degu}
        $\ex[\de(u_{r})]=O(\frac{1}{l}\sum_{e \in \eun} w(e))$.
    \end{claim}

    \begin{proof}[Proof of Claim \ref{cl:total}]
        As described in App.\ \ref{Data Sets},
        we can pick a random uncolored edge in constant time. 
        In addition, 
        we can pick a random missing color on a vertex in (expected) time linear to its degree. Also, we can compute the path $P_{r}$ in $O(|P_{r}|)$ time by repeatedly adding edges to it while possible to do.
        Consequently, by Lemmas \ref{making fan bound} and \ref{Extend bound}, which bound the running time of Algorithms \texttt{Make-Primed-Fan} and \texttt{Extend-Coloring} by $O(|P_{r}|)+O(\de(u_{r}))$, we conclude that $\ex[T(\texttt{Color-One-Edge})] = O(\ex[|P_{r}|] +\ex[\de(u_{r})])$.
    \end{proof}
    
    \begin{proof}[Proof of Claim \ref{cl:lengthP}]
        By Observation \ref{Path Bound} $\ex[|P_{r}|] \le \ex[|I(P_{r})|]+O(1)$. We next prove  that
        $\ex[|I(P_{r})|]=O(\frac{1}{l}\sum_{e \in \ecl} w(e))$.

        For every maximal alternating path $P$ in $MP$, 
        let $u_{0}(P)$ and $u_{|P|}(P)$ be the first and last endpoints of $P$, respectively, and let $c_{0}(P)$ and $c_{|P|}(P)$ be the missing colors of $u_{0}(P)$ and $u_{|P|}(P)$ from the two colors of $P$, respectively.
        For a vertex $v$, denote by $l(v)$ the number of uncolored edges incident on $v$.
        Note that for every $P$ in $MP$,
        \begin{eqnarray*}
        \pr(P_{r}=P)&\leq& \pr(c_{r}=c_{0}(P)  ~\vert~ u_{r} = u_{0}(P)) \cdot \pr(u_{r} =u_{0}(P))\\
        && ~+~ \pr(c_{r}=c_{|P|}(P) ~\vert~ u_{r} =u_{|P|}(P)) \cdot \pr(u_{r} = u_{|P|}(P))\\
        &\leq& \frac{1}{|M(u_{0}(P))|}\cdot\frac{l(u_{0}(P))}{l} ~+~
        \frac{1}{|M(u_{|P|}(P))|}\cdot\frac{l(u_{|P|}(P))}{l}\\
         &\leq& \frac{1}{|M(u_{0}(P))|}\cdot\frac{|M(u_{0}(P))|}{l} ~+~ \frac{1}{|M(u_{|P|}(P))|}\cdot\frac{|M(u_{|P|}(P))|}{l} ~=~\frac{2}{l}.
        \end{eqnarray*}
        It follows that
        \[\ex{[|I(P_{r})|]} ~=~
        \sum_{P\in MP} \pr(P_{r}=P) \cdot |I(P)| ~\leq~
        \sum_{P\in MP} \frac{2}{l} \cdot |I(P)| 
        ~~~~~~~\myineq~~~~~~~
         \sum_{e \in \ecl} \frac{2}{l} \cdot w(e)\]
         which is $ O\left(\frac{1}{l}\sum_{e \in \ecl} w(e)\right)$.
    \end{proof}
    
    \begin{proof}[Proof of Claim \ref{cl:degu}]
        Define $D_{m}(u)$ to be the set of uncolored edges $e=(u,v)$, such that 
         $\de(u)\le \de(v)$.
        Now,
        $$\pr(u_{r}= u)\\ ~\leq~ \pr\left(e_{r} =(u,v) \cap \de(u) \le \de(v)\right) 
        ~=~ \frac{|D_{m}(u)|}{l}.$$
        Consequently, \begin{eqnarray*}
        \ex[\de(u_{r})] &=&
        \sum_{u \in V} \pr(u_{r}=u)\cdot \de(u)
        ~\leq~ \sum_{u \in V} \frac{|D_{m}(u)|}{l}\cdot \de(u) ~=~
        \frac{1}{l}\sum_{u \in V} \sum_{e \in D_{m}(u)} \de(u) 
        \\&=&
        \frac{1}{l}\sum_{e \in \eun}\sum_{x: e \in D_{m}(x)}\de(x) ~\le~
        \frac{1}{l}\sum_{e=(u,v) \in \eun} 2\cdot \min\{\de(u),\de(v)\}
        \end{eqnarray*}
        which is $O\left(\frac{1}{l}\sum_{e \in \eun} w(e)\right)$.
    \end{proof}

    Now we are ready to complete the proof of Lemma \ref{Color-One-Edge Bound}.
    Using Claims \ref{cl:total}, \ref{cl:lengthP} and \ref{cl:degu} we get 
    \begin{eqnarray*}
    \ex[T(\texttt{Color-One-Edge})] ~=~ O(\ex[|P_{r}|]+ \ex[\de(u_{r})])
    ~=~ O\left(1 + \frac{1}{l}\sum_{e \in \ecl} w(e) + \frac{1}{l}\sum_{e \in \eun} w(e)\right)
    \end{eqnarray*}
    which is $O\left(\frac{w(G)}{l}\right)$.
    Recalling that $w(G) = O(m \alpha)$ holds by Claim \ref{cl:basearb} completes the proof.
\end{proof}

\subsection{Algorithm \texttt{Color-Edges}: Coloring a Single Edge Iteratively}

The following algorithm is given as input a graph $G$ and a partial $(\Delta+1)$-coloring $\chi$, and it returns as output a proper $(\Delta+1)$-coloring for $G$. This algorithm is identical to the analogous one by \cite{Sin}; however, here we invoke our modified \texttt{Color-One-Edge} Algorithm in line 2 rather than the analogous one by \cite{Sin}.

\begin{algorithm}[H]\label{alg:many}
    \SetAlgoLined
    \DontPrintSemicolon
    \KwIn{A graph $G$ with a partial ($\Delta+1$)-coloring $\chi$}
    \KwOut{Updated partial ($\Delta+1$)-coloring $\chi$ of $G$, such that all the edges in $G$ are colored}
    \While{$\eun \ne \emptyset$} {
    $\chi \leftarrow \texttt{Color-One-Edge}(G,\chi)$\;
    }
    \Return{$\chi$}
    \caption{\texttt{Color-Edges}$(G,\chi)$}
\end{algorithm}

\begin{lemma}
\label{Color-Edges Bound}
    The expected runtime of Algorithm \texttt{Color-Edges} on a graph $G$ with an empty partial ($\Delta+1$)-coloring $\chi$ is $O(w(G)\log n) = O(m \alpha \log n)$.
\end{lemma}

\begin{proof}
    As described in App.\ \ref{Data Sets}, we can check whether $\eun = \emptyset$ in constant time.
    Since each call to Algorithm $\texttt{Color-One-Edge}$ colors a single uncolored edge, the while loop consists of  $m$ iterations. Also, the runtime of each iteration is that of the respective call to \texttt{Color-One-Edge}.
    At the beginning of the $i$th iteration, the number of uncolored edges $l$ is $m-i+1$,
    so by Lemma \ref{Color-One-Edge Bound}, the expected runtime of the $i$th iteration, denoted by
    $\ex[T(i \mbox{th iteration})]$,
    is $O(\frac{w(G)}{m-i+1})$.
    It follows that
    $$\ex[T(\texttt{Color-Edges)}] ~=~   \sum_{i=1}^{m}\ex[T(i \mbox{th iteration})]~=~ O\left(\sum_{i=1}^{m}\frac{w(G)}{m-i+1}\right)~=~  O(w(G)\log n).
   $$
   Recalling that $w(G) = O(m \alpha)$ holds by Claim \ref{cl:basearb} completes the proof.
\end{proof}

\section{Our $(\Delta+1)$-Coloring Algorithm: \texttt{Recursive-Color-Edges}} \label{sec:advanced}

In this section we present our $(\Delta+1)$-edge-coloring algorithm,
\texttt{Recursive-Color-Edges},
which proves Theorem \ref{th:main}. Our algorithm is similar to that of \cite{Sin}, which, in turn, in based on Gabow et al.\ \cite{Gabow85} --- except for a few small yet crucial tweaks, one of which is that we employ our algorithm \texttt{Color-Edges} as a subroutine
rather than the analogous subroutine from \cite{Sin}. 
We first describe the approach taken by \cite{Gabow85,Sin},
and then present our algorithm and emphasize the specific modifications needed for achieving the improvement in the running time. 

\subsection{The Approach of \cite{Gabow85,Sin}}

The algorithm of \cite{Gabow85,Sin} employs a natural divide-and-conquer approach. It partitions the input graph into
two edge-disjoint subgraphs of maximum degree roughly $\Delta/2$, then it recursively computes a coloring with at most $\Delta/2 + 2$ colors for each subgraph separately, and then it stitches together the two colorings into a single coloring. 
Naively
stitching the two colorings  into one would result in up to $\Delta+4$ colors, so the idea is to prune excessive colors and then deal with the remaining uncolored edges via a separate coloring algorithm. 

In more detail, the algorithm consists of four phases: \textbf{Partition}, \textbf{Recurse}, \textbf{Prune}, and \textbf{Repair}.

\begin{itemize}
    \item \textbf{Partition.} The algorithm partitions the edges of the graph into two edge-disjoint subgraphs, such that the edges incident on each vertex are divided between the two subgraphs almost uniformly. This in particular implies that the maximum degree in each subgraph is roughly $\Delta/2$.
    Such a partition can be achieved by a standard procedure, \EMPH{Euler partition},  
    which was used also by \cite{Gabow85,Sin}.
    For completeness,
    in Section \ref{Euler Partitioning}
    we describe this procedure and  prove some basic properties that will be used later. 
    \item \textbf{Recurse.} The algorithm recursively computes a coloring with at most $\Delta/2+2$ colors for each subgraph separately, where the two colorings use disjoint palettes of colors. Then, we combine the two colorings into one by simply taking their union,
    which results with a proper coloring with at most $\Delta+4$ colors.
    \item \textbf{Prune.} At this point, the number of colors used in the coloring is $\Delta'$, for $\Delta'  \le \Delta+4$, which exceeds the required bound of $\Delta+1$. To prune the up to three extra colors,
    the algorithm 
    groups the edges into color classes, 
    chooses the $\Delta' - (\Delta + 1)$ color classes of smallest size,
    and then uncolors all edges in those color classes. 
    {\em In our algorithm, we choose the  $\Delta' - (\Delta + 1)$ color classes of smallest  \EMPH{weight}, where the weight of a color class is the sum of weights of edges with that color.}
    \item \textbf{Repair.} To complete the partial coloring into a proper $(\Delta+1)$-coloring, each of the uncolored edges resulting from the \textbf{Prune} phase has to be recolored; this is done by a separate coloring algorihtm. {\em In our algorithm, this separate coloring algorithm is Algorithm \texttt{Color-Edges} from Section \ref{sec:basicalg}.}

\end{itemize}

\subsection{The \texttt{Euler Partition} Procedure}
\label{Euler Partitioning}

An \EMPH{\texttt{Euler partition}} is a partition of the edges of a graph into a set of edge-disjoint tours,
such that every odd-degree vertex is the endpoint of exactly one tour, and no even-degree vertex
is an endpoint of any tour (some tours may be cycles). Such a partition can be computed greedily in linear time, by simply removing maximal tours of the graph until no edges remain. The edges of the graph can then be split between the subgraphs by traversing each tour and alternately assigning the edges to the two subgraphs.
Internal edges of the tours split evenly between the two subgraphs, so only two cases may cause an unbalanced partition:

\begin{itemize}
    \item \textbf{Case 1: The two endpoints of any tour.} The only edge incident on any tour endpoint is assigned to one of the two subgraphs, which causes at most 1 extra edge per vertex in one of the two subgraphs.
    \item \textbf{Case 2: A single vertex in any tour that is an odd length cycle.} The cycle edges can be split evenly among all vertices of the cycle but one, the {\em starting vertex}, for which there are 2 extra edges in one of the two subgraphs. Note that the algorithm is free to choose (1) any of the cycle vertices as the starting vertex, and (2) the subgraph among the two to which the 2 extra edges would belong;
the algorithm will carry out these choices so as to minimize the discrepancy in degrees of any vertex.
\end{itemize}
The next observation follows directly from the description of the \texttt{Euler partition} procedure (and by handling the two aforementioned cases that may cause an unbalanced partition in the obvious way). 
\begin{observation}
\label{Basic Property}
    Let $G$ be a graph and let $G_{1},G_{2}$ be the subgraphs of $G$ obtained by the \texttt{Euler partition} procedure.
    Then for every vertex $v\in G$,
    the degrees of $v$ in $G_{1}$ and $G_{2}$, denoted by $\de_{G_{1}}(v)$ and $\de_{G_{2}}(v)$ respectively, satisfy
    $$\frac{1}{2}\de_{G}(v)-1 ~\leq~ \de_{G_{1}}(v),\de_{G_{2}}(v) ~\leq~ \frac{1}{2}\de_{G}(v)+1.$$
\end{observation}

\subsubsection{Analysis of the \texttt{Euler Partition} procedure}
\label{Euler Partitioning.1}
Let $G$ be an $n$-vertex $m$-edge graph, and consider any algorithm $A$ that applies the \texttt{Euler Partition} procedure recursively.
That is, upon recieving the graph $G$ as input, Algorithm $A$ 
splits $G$ into two subgraphs $G_{1}$ and $G_{2}$ {\em using the \texttt{Euler Partition} procedure}, and then recursively applies Algorithm $A$ on $G_{1}$ and   $G_{2}$.

Consider the binary recursion tree
of algorithm $A$:
The first \EMPH{level $L_{0}$} consists of the root node that corresponds to the graph $G$.
The next \EMPH{level $L_{1}$} consists of two nodes that correspond to the two subgraphs of $G$ obtained by applying the \texttt{Euler Partition} procedure on $G$.
In general, the \EMPH{$i$th level $L_{i}$} consists of $2^{i}$ nodes that correspond to the $2^i$ subgraphs of $G$ obtained by applying the \texttt{Euler Partition} procedure on the $2^{i-1}$ subgraphs of $G$ at level $L_{i-1}$.
In App.~\ref{app:proofs}, we prove some basic properties of the recursion tree of Algorithm $A$.

\subsection{The Algorithm}

In this section we present our $(\Delta + 1)$-coloring algorithm.
As mentioned, our algorithm follows closely that of  \cite{Gabow85,Sin}, but we introduce the following changes:
(1) In the \textbf{Prune} phase, we uncolor edges from the three color classes of minimum weight (rather than minimum size); (2) in the \textbf{Repair} phase, as well as at the bottom level of the recursion, we invoke our modified \texttt{Color-Edges} algorithm rather than the analogous one by \cite{Sin,Gabow85}. \texttt{Recursive-Color-Edges} gives the pseudocode for our algorithm.

\begin{algorithm}[H]\label{alg:static}
    \SetAlgoLined
    \DontPrintSemicolon
    \KwIn{A graph $G$}
    \KwOut{A ($\Delta+1$)-coloring $\chi$ of $G$}
    \If{$\Delta \leq 2\sqrt{\frac{n}{\log n}}$}{
        $\chi \leftarrow \mbox{Empty $(\Delta+1)$-Coloring of }G$\;
        $\chi \leftarrow \texttt{Color-Edges}(G,\chi)$\;
        \Return{$\chi$}
    }
    \textbf{Partition:}\\
    Decompose $G$ into subgraphs $G_{1}$ and $G_{2}$ by applying the \texttt{Euler partition} procedure.\;
    \textbf{Recurse:} // color $G_{1}$ and $G_{2}$ by at most $\Delta+4$ colors\;
        $\chi_{1} \leftarrow \mbox{Empty $(\Delta_{G_{1}}+1)$-coloring of }G_{1}$ // $\Delta_{G_1} = \Delta(G_1)$ \\
        $\chi_{2} \leftarrow \mbox{Empty $(\Delta_{G_{2}}+1)$-coloring of }G_{2}$ // $\Delta_{G_2} = \Delta(G_2)$ \\
        $\chi_{1} \leftarrow \texttt{Recursive-Color-Edges}(G_{1},\chi_{1})$\;
        $\chi_{2} \leftarrow \texttt{Recursive-Color-Edges}(G_{2},\chi_{2})$\;
        $\chi \leftarrow$ $(\Delta+4)$-coloring of $G$
        obtained by merging $\chi_{1}$ and $\chi_{2}$\;
    \textbf{Prune:}\\
        \While{there are more than $\Delta+1$ colors in $\chi$}{
        $c_{l} \leftarrow$ a color that minimizes $w(c_{l}):=\sum_{e \in E: e \mbox{ colored by } c_{l}}w(e)$\;
        Uncolor all edges colored by $c_{l}$ // The color $c_{l}$ is removed from $\chi$\;
        }
     \textbf{Repair:} // color all   uncolored edges\;
        $\chi \leftarrow \texttt{Color-Edges}(G,\chi)$\;
        \Return{$\chi$}
    \caption{\texttt{Recursive-Color-Edges}$(G)$}
\end{algorithm}

\subsection{Analysis of The Algorithm}

We consider the binary recursion tree of
Algorithm \texttt{Recursive-Color-Edges},
and note that this algorithm can assume the role of Algorithm $A$ in Section \ref{Euler Partitioning.1};
in particular, we follow the notation of Section \ref{Euler Partitioning.1}.

\begin{lemma} \label{lem:timesubgraph}
    For any subgraph $H=(V_{H},E_{H})$ of $G$ at level $L_{i}$ with $m_{H}$ edges,
    maximum degree $\Delta_{H}$ and weight $W_{H}$,
    the expected runtime spent on $H$ due to the call to \texttt{Recursive-Color-Edges} on $G$ is bounded by $$O\left(m_{H}+\frac{W_{H}}{\Delta_{H}}\cdot \left(\frac{n}{\Delta_{H}}+\log n\right)\right).$$
\end{lemma}

\begin{proof}
    Consider first the case $\Delta_{H}\leq 2\sqrt{\frac{n}{\log n}}$.
    The time of the creation of the empty coloring (in line 2 of the code) is $O(m_{H})$. 
    By Lemma \ref{Color-Edges Bound}, the expected time spent on $H$ while running \texttt{Color-Edges} (in line 3) is bounded by
    $$O(W_{H}\log n) ~=~O\left(m_{H}+W_{H}\frac{n}{\Delta_{H}^{2}}\right)\\
    ~=~ O\left(m_{H}+\frac{W_{H}}{\Delta_{H}}\cdot\frac{n}{\Delta_{H}}\right)$$
    $$=O\left(m_{H}+\frac{W_{H}}{\Delta_{H}}\cdot \left(\frac{n}{\Delta_{H}}+\log n\right)\right).$$
    We may henceforth assume that $\Delta_{H}>2\sqrt{\frac{n}{\log n}}$.
    Next, we analyze the time required by every phase of the algorithm.
    We first note that the first three phases of the algorithm can be implemented in $O(m_H)$ time:

\begin{itemize}
    \item \textbf{Partition:}  The \texttt{Euler partition} procedure takes $O(m_{H})$ time,
 as explained in Section \ref{Euler Partitioning}.
    \item \textbf{Recurse:}  We only consider the time spent at the \textbf{Recurse} phase on $H$ itself, i.e.,  the time needed to create the two empty colorings $\chi_1$ and $\chi_2$ and the time needed to merge them into $\chi$, each of which takes time $O(m_{H})$.

    \item \textbf{Prune:} In $O(m_H)$ time we can scan all edges and group them into color classes, compute the weight of each color class, and then find the up to three color classes of lowest weight. The same amount of time  suffices for uncoloring all edges in those three color classes, thereby removing those colors from $\chi$.
\end{itemize}

\noindent
It remains to bound the time required for the \textbf{Repair} phase, denoted by $T(Repair\,H)$.
We will prove that 
\begin{equation} \label{eq:repair}
\ex[T(Repair\,H)]~=~O\left(\frac{W_{H}}{\Delta_{H}}\cdot \left(\frac{n}{\Delta_{H}}+\log n\right)\right),
\end{equation}
and conclude that total expected time of the algorithm is $$O\left(m_{H}+\frac{W_{H}}{\Delta_{H}}\cdot \left(\frac{n}{\Delta_{H}}+\log n\right)\right).$$
We shall bound the expected time for coloring the uncolored edges via Algorithm $\texttt{Color-Edges}$, where 
    the uncolored edges are the ones that belong to the three color classes of minimum weight (We may assume w.l.o.g.\ that exactly three colors have been uncolored in $H$, out of a total of $\Delta_H+4$ colors, by simply adding dummy color classes of weight 0).
    By an averaging argument, the total weight of the uncolored edges is bounded by
\begin{equation} \label{eq:uncoloredwt} \frac{3}{\Delta_{H}+4}\cdot W_{H}=O\left(\frac{W_{H}}{\Delta_{H}}\right).
\end{equation}
Let $PCL(H)$ be the set of all possible partial $(\Delta_{H}+1)$-colorings of $H$, and for every coloring $\chi \in PCL(H)$ let $U(\chi)$ denote the number of edges of $H$ that are uncolored by $\chi$.
    In addition, let $PCL(H,l)$ be the set of all colorings $\chi \in PCL(H)$ with $U(\chi)=l$.

    Note that the partial coloring obtained at the beginning of the \textbf{Repair} phase is not deterministic, and let $\chi_{0}$ denote this random partial coloring. Thus, $U=U(\chi_{0})$ is a random variable.
    Fix an arbitrary integer $l \ge 0$.
    Under the condition $U=l$, Algorithm \texttt{Color-Edges} consists of $l$ iterations that color uncolored edges via Algorithm $\texttt{Color-One-Edge}$.
    
    For every iteration $i=1,...,l$, let $\chi_{i}\in PCL(H,l+1-i)$ be the random partial coloring at the beginning of the iteration, let $e_{i}=(u_{i},v_{i})$ be the random uncolored edge chosen in line 1 of Algorithm $\texttt{Color-One-Edge}$, with $\de(u_{i})\le \de(v_{i})$, let $c_{i}$ be the random missing color of $u_{i}$ chosen in line 4 of that algorithm, and let $P_{i}$ be the maximal alternating path starting at $u_{i}$ obtained in line 5 of that algorithm.

    Each uncolored edge $e_{i}$ is colored via a call to Algorithm \texttt{Color-One-Edge}, whose expected time is dominated (by Claim \ref{cl:total}) by $O(\ex[|P_{i}|] + \ex[\de(u_{i})])$.
    So the total expected time for coloring all the $l$ uncolored edges under the condition $U=l$, namely $\ex[T(Repair\,H)~|~U=l]$, satisfies
    \begin{eqnarray}
    \label{total repair}
    \ex[T(Repair\,H)~|~U=l] ~=~ 
         \sum_{i=1}^{l} O(\ex[|P_{i}|~|~U=l]) + \sum_{i=1}^{l} O(\ex[\de(u_{i})~|~U=l]).
    \end{eqnarray}
    In App.~\ref{app:proofs}, we prove the following two claims, and use them to prove Eq.\ \ref{eq:repair}.

    \begin{claim}
    \label{low fans bound}
        For any integer $l \ge 0$,
        $\sum_{i=1}^{l} \ex[\de(u_{i})~|~U=l] = O\left(\frac{W_{H}}{\Delta_{H}}\right)$.
    \end{claim}

    \begin{claim}
    \label{low paths bound}
        For any integer $l \ge 0$, $\sum_{i=1}^{l}\ex[|P_{i}|\mbox{ $|$ } U=l] = O\left(\frac{W_{H}}{\Delta_{H}}\cdot\left(\frac{n}{\Delta_{H}}+\log n\right)\right)$.
    \end{claim}
\end{proof}

In App.~\ref{app:proofs} we prove the following lemma using Lemma~\ref{lem:timesubgraph}.
\begin{lemma}\label{lem:last}
    The expected runtime of Algorithm \texttt{Recursive-Color-Edges} on $G$ is bounded by
    $$O\left(W\cdot \min\left\{\log n,\frac{\sqrt{n\log n}}{\Delta}\right\}+m\log n\right).$$
\end{lemma}
\noindent
{\bf Remark.}
Claim \ref{cl:basearb} yields $W = O(m \alpha)$, hence the expected runtime of the algorithm is $O(\alpha m\cdot \min\{\log n,\frac{\sqrt{n\log n}}{\Delta}\}+m\log n)$,
or equivalently, 
    $O(\min\{m \Delta \cdot \log n,m\sqrt{n\log n}\}\cdot \frac{\alpha}{\Delta}+m\log n)$.
It is straightforward to achieve the same bound on the running time (up to a logarithmic factor) with high probability. (As mentioned, we do not attempt to optimize polylogarithmic factors in this work.) 
Thus we derive the following corollary, which concludes the proof of Theorem \ref{th:main}.
\begin{corollary}
    One can compute a $(\Delta+1)$-coloring in any $n$-vertex $m$-edge graph of arboricity $\alpha$ and maximum degree $\Delta$ within a high probability runtime bound of
    $$O\left(\min\{m \Delta \log^{2} n,m \sqrt{n}\log^{1.5} n\}\cdot \frac{\alpha}{\Delta}+m\log^{2} n\right).$$
\end{corollary}

\newpage
\bibliography{arb}

\newpage
\appendix
\section{Figures}
\label{app:figs}

\begin{figure}[h]
    \centering
    \includegraphics[width=0.35\textwidth]{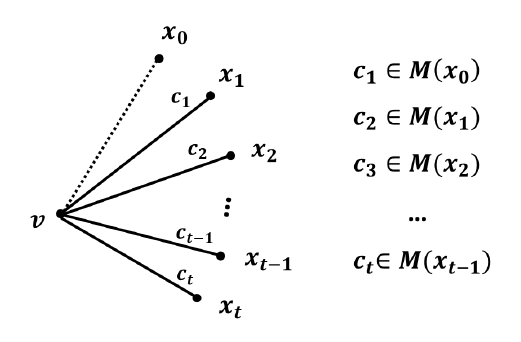}
    \caption{An illustration of a fan}
    \label{fig:fan}
\end{figure}

\begin{figure}[h]
    \centering
    \includegraphics[width=0.5\textwidth]{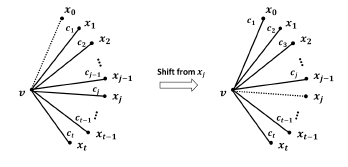}
    \caption{An illustration of fan shifting}
    \label{fig:fanShift}
\end{figure}

\begin{figure}[h]
    \centering
    \includegraphics[width=0.5\textwidth]{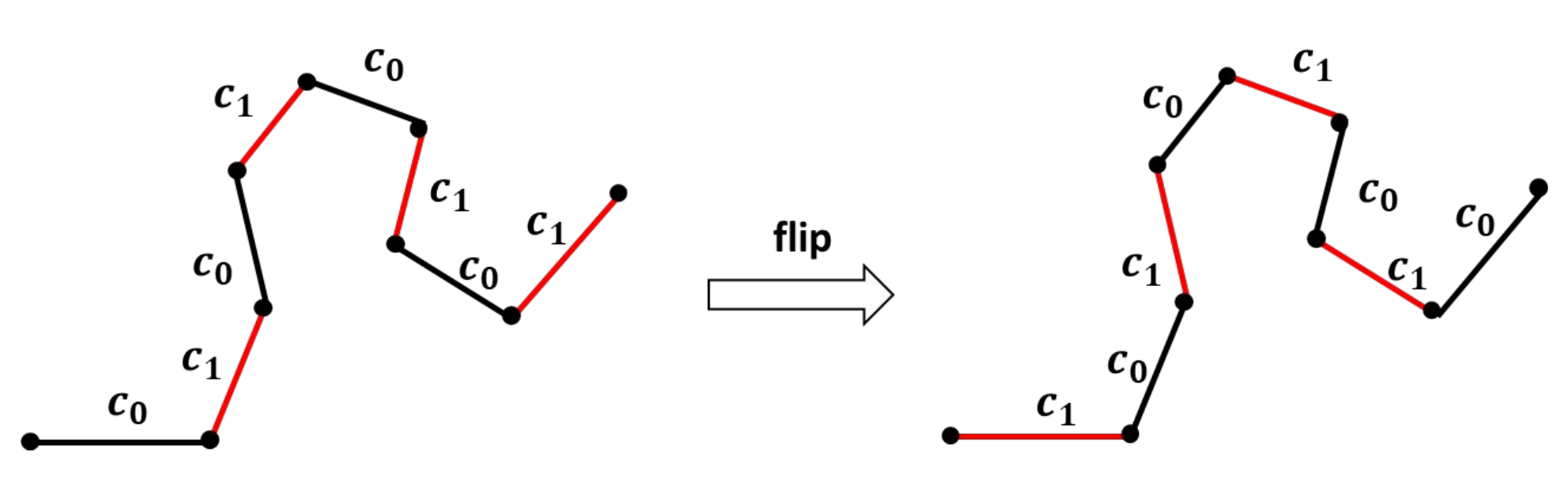}
    \caption{An illustration of path flipping}
    \label{fig:pathFlip}
\end{figure}

\section{Data Structures}
\label{Data Sets}
In this appendix we briefly describe the data structures used by our algorithms; it is readily verified that the initialization time of the data structures as well as their space usage is linear in the graph size.

    First, for every (non-isolated) vertex $v$ we maintain a hash table of size $O(\de(v))$ that contains, for every color $c \notin M(v)$, a pointer to the edge incident to $v$ colored $c$. 
    Clearly, $O(1)$ (expected) time suffices for checking whether or not any color $c$ is missing on $v$ and for finding the edge that occupies color $c$ if $c\notin M(v)$, as well as for updating the hash table following a color $c \in M(v)$ that leaves $M(v)$ or vice versa.

    To quickly find a missing color on any (non-isolated) vertex $v$, we  maintain a non-empty list of all the colors in $M(v)\cap [\de(v)+1]$, accompanied with an array of size $\de(v)+1$ of the color palette $[\de(v)+1]$, with mutual pointers between the two. 
    Clearly, $O(1)$ time suffices for finding an arbitrary missing color on $v$ (via the list) and for updating the list and the corresponding array due to a coloring and uncoloring of any edge incident to $v$.

    Next, we argue that a random missing color on any vertex $v$ can be found in expected time $O(\de(v))$. 
     If $\de(v)>\Delta/2$, we can create an auxiliary array of all the missing colors of $v$ in $O(\Delta)=O(\de(v))$ time and sample a random color from the auxiliary array.
     In the complementary case $\de(v)\leq \Delta/2$, we can just sample a random color from the entire color palette $[\Delta+1]$ repeatedly, until sampling a color that is missing on $v$.
     With probability at least $1/2$ we sample a missing color on $v$, so the expected time of the process is $O(1)$.

        Finally, we maintain the set $\eun$ of all uncolored edges (and a variable holding its size $|\eun|$) via an array of size $m$, where the first $|\eun|$ entries of the array hold pointers to the uncolored edges, and we also maintain, for each uncolored edge $e$, a variable $ind(e)$ holding its respective index in the array. We can determine the number of uncolored edges in $O(1)$ time
        via the variable $|\eun|$. When an uncolored edge $e$ gets colored, we first copy to  position $ind(e)$ of the array the pointer to the uncolored edge $e'$ corresponding to position $|\eun|$ of the array, then update the respective index $ind(e')$ of edge $e'$ to $ind(e)$, 
        and finally decrement the variable $|\eun|$. When a colored edge $e$ gets uncolored, we first increment the variable $|\eun|$, then put in position $|\eun|$ of the array a pointer to edge $e$, and finally set the  index $ind(e)$ of edge $e$ to be $|\eun|$. 
All of this clearly takes $O(1)$ time.
Using the array and the variable $|\eun|$, we can  pick a random uncolored edge in $O(1)$ time.

\section{Deferred Proofs from Section~\ref{sec:advanced}}\label{app:proofs}

This appendix contains the missing proofs from from Section~\ref{sec:advanced}.

\subsection{Properties of \texttt{Euler Partition}}

Recall that $\Delta = \Delta(G)$ denotes the maximum degree of $G$ and let $W =  w(G) =\sum_{e \in E} w(e)$
denote the weight of $G$.
We also define:
$\Delta_{i}:=2^{-i}\Delta$ and $W_{i}:=2^{-i}W$.
For a subgraph $H$ of $G$, let $w_H$
be the weight function of $H$, i.e., $w_H(e) = \min\{\de_{H}(u),\de_{H}(v)\}$ is the minimum degree in $H$ of the two endpoints of any edge $e = (u,v)$ in $H$,
and let
$w(H) =\sum_{e \in H} w_H(e)$
denote the weight of $H$.

\begin{claim}
\label{degrees}
    For any level $L_i$, any subgraph $H$ of $G$ at level $L_{i}$ and 
    any vertex $v \in G$, the degree $\de_{H}(v)$ of $v$ in $H$ satisfies    $2^{-i}\cdot \de_{G}(v)-2 \leq \de_{H}(v)
    \leq 2^{-i}\cdot \de_{G}(v)+2$.
\end{claim}

\begin{proof}
    We will prove the following stronger claim by induction on $i$:
$$2^{-i}\cdot \de_{G}(v)-\sum_{k=1}^{i}2^{1-k} ~\leq~ \de_{H}(v) ~\leq~ 2^{-i}\cdot \de_{G}(v)+\sum_{k=1}^{i}2^{1-k}.$$

For the induction basis $i=0$, we have $H = G$ and
$$2^{-i}\cdot \de_{G}(v)-\sum_{k=1}^{i}2^{1-k} ~=~ \de_G(v) ~=~ \de_{H}(v) ~=~ 2^{-i}\cdot \de_{G}(v)+\sum_{k=1}^{i}2^{1-k}.$$

For the induction step, we assume that the claim holds for every subgraph at level $L_{i-1}$, and prove that it holds for every subgraph $H$ at level $L_{i}$.
Let $H_{i-1}$ be the subgraph at level $L_{i-1}$, such that $H$ is one of the two subgraphs of $H_{i-1}$ at level $L_i$.
By Observation \ref{Basic Property},
$\frac{1}{2}\de_{H_{i-1}}(v)-1\leq \de_{H}(v)\leq \frac{1}{2}\de_{H_{i-1}}(v)+1.$
By the induction hypothesis, 
$$2^{-i+1}\cdot \de_{G}(v)-\sum_{k=1}^{i-1}2^{1-k} ~\leq~ \de_{H_{i-1}}(v) ~\leq~ 2^{-i+1}\cdot \de_{G}(v)+\sum_{k=1}^{i-1}2^{1-k}.$$
It follows that
$$\frac{1}{2}\left(2^{-i+1}\cdot \de_{G}(v)-\sum_{k=1}^{i-1}2^{1-k}\right)-1 ~\leq~ \de_{H}(v) ~\leq~ \frac{1}{2}\left(2^{-i+1}\cdot \de_{G}(v)+\sum_{k=1}^{i-1}2^{1-k}\right)+1,$$
and we get that
$$2^{-i}\cdot \de_{G}(v)-\sum_{k=1}^{i}2^{1-k} ~\leq~ \de_{H}(v) ~\leq~ 2^{-i}\cdot \de_{G}(v)+\sum_{k=1}^{i}2^{1-k}.$$
\end{proof}

\begin{claim} \label{cl:DeltaH}
    For any level $L_i$ and  any subgraph $H$ of $G$ at level $L_{i}$, 
    the maximum degree  $\Delta_{H}$ of $H$ satisfies
    $\Delta_{i}-2 \leq \Delta_{H} \leq \Delta_{i}+2$.
\end{claim}

\begin{proof}
    Let $v$ be a vertex with maximum degree in $H$. By Claim \ref{degrees},  $$\Delta_{H} ~=~ \de_{H}(v) ~\leq~ 2^{-i}\cdot \de_{G}(v)+2 ~\leq~ 2^{-i}\cdot \Delta+2 ~=~ \Delta_{i}+2.$$
Now let $u$ be a vertex with maximum degree in $G$. Applying Claim \ref{degrees} again, we obtain
$$\Delta_{H} ~\geq~ \de_{H}(u) ~\geq~ 2^{-i}\cdot \de_{G}(u)-2=2^{-i}\cdot \Delta-2 ~=~ \Delta_{i}-2.$$
It follows that $\Delta_{i}-2\leq \Delta_{H}\leq \Delta_{i}+2$.
\end{proof}

\begin{claim} \label{cl:boundwt}
    For every level $L_{i}$: $\sum_{H\in L_{i}}w(H)\leq W_{i}+2m.$
\end{claim}

\begin{proof}
    \begin{eqnarray*}
    \sum_{H\in L_{i}}w(H)&=&\sum_{H\in L_{i}}\sum_{e\in H}w_{H}(e)
    ~=~\sum_{H\in L_{i}}\sum_{(u,v)\in H}\min\{\de_{H}(u),\de_{H}(v)\}
    \\&\mydineq&~~~\sum_{H\in L_{i}}\sum_{(u,v)\in H}(2^{-i}\min\{\de_{G}(u),\de_{G}(v)\}+2),
\end{eqnarray*}
Noting that every edge $e$ in $G$ appears in exactly one subgraph of $G$ at level $L_{i}$, it follows that
\begin{eqnarray*}
 \sum_{H\in L_{i}}w(H) 
    &\leq&\sum_{H\in L_{i}}\sum_{(u,v)\in H}(2^{-i}\min\{\de_{G}(u),\de_{G}(v)\}+2)\\
    &=&\sum_{(u,v)\in G}(2^{-i}\min\{\de_{G}(u),\de_{G}(v)\}+2)\\
    &=&\sum_{e\in G}(2^{-i}w_{G}(e)+2)
    ~=~2^{-i}\sum_{e\in G}w_{G}(e)+2m ~=~ W_{i}+2m.
\end{eqnarray*}
\end{proof}

\subsection{Proof of Claim~\ref{low fans bound}}

For edge $e_{i}=(u_{i},v_{i})$,  $u_i$ is chosen as it satisfies $\de(u_i) \le \de(v_i)$, so
$\de(u_{i})=w(e_i)$.
To complete the proof, we note that Eq.\ \ref{eq:uncoloredwt} implies that the total weight of the uncolored edges is bounded (deterministically) by $O\left(\frac{W_{H}}{\Delta_{H}}\right)$.

\subsection{Proof of Claim~\ref{low paths bound}}

We say that a vertex $v$ has \EMPH{high degree} if $\de(v) \geq \Delta_{H}/2$, and it has \EMPH{low degree} otherwise.
        Similarly, we say that edge $e$ has \EMPH{high weight} if $w(e) \geq \Delta_{H}/2$ (i.e. both its endpoints have high degree), and it has \EMPH{low weight} otherwise.
        We need to bound $\ex \left[\sum_{i=1}^{l}|P_{i}|\mbox{ $|$ } U=l\right]$.
        We will bound the expected sum of the lengths of paths of high weight edges (i.e., paths $P_{i}$ such that the corresponding edge $e_{i}$ has high weight) separately from the expected sum of the lengths of paths of low weight edges (i.e., paths $P_{i}$ such that $e_{i}$ has low weight); in fact, for high weight edges, the upper bound  will hold deterministically. 

        \paragraph{High weight edges.}
         Eq.\ \ref{eq:uncoloredwt} implies that the total weight of the uncolored edges is bounded by $O\left(\frac{W_{H}}{\Delta_{H}}\right)$,
        hence the number of uncolored edges of high weight is bounded by
        $$O\left(\frac{W_{H}}{\Delta_{H}}\cdot \frac{1}{\Delta_{H}/2}\right)~=~
        O\left(\frac{W_{H}}{\Delta_{H}^{2}}\right).$$
        Since each chosen alternating path is simple, and as such has at most $n-1$ edges, we can bound the sum of the lengths of the paths of high weight edges (deterministically) by
        $$O\left(n \cdot \frac{W_{H}}{\Delta_{H}^{2}}\right) ~=~
        O\left(\frac{W_{H}}{\Delta_{H}}\cdot\frac{n}{\Delta_{H}}\right) ~=~
        O\left(\frac{W_{H}}{\Delta_{H}}\cdot\left(\frac{n}{\Delta_{H}}+\log n\right)\right).$$

    \paragraph{Low weight edges.}
        
        Let $S_{H}$ be the random variable given by the sum of lengths of the paths of the low weight edges. We have to bound $\ex[S_{H}~|~U=l]$.
        
        For every $1 \le i \le l$, let $X_{i}$ be the random indicator variable that gets value 1 iff the edge colored at the $i$th iteration has low weight, and define  the random variable 
        $Y_{i} :=|I(P_{i})|\cdot X_{i}$. 
        
        Observe that
        \begin{eqnarray} \label{eq:sh}
            \nonumber \ex\left[S_{H}~|~U=l\right]&=& \ex\left[\sum_{1\le i\le l: X_{i}\neq 0}|P_{i}|~\middle|~U=l\right]
            ~=~\ex\left[\sum_{i=1}^{l}|P_{i}|\cdot X_{i}~\middle|~U=l\right]\\ \nonumber
            &=&\sum_{i=1}^{l}\ex\left[|P_{i}|\cdot X_{i}~|~U=l\right]
            ~~~~~~~~~~\myobeq~~~~~~~~~~ \sum_{i=1}^{l}\ex[(|I(P_{i})|+O(1))\cdot X_{i}~|~U=l]\\ 
            &=&\sum_{i=1}^{l}(\ex[|I(P_{i})|\cdot X_{i}~|~U=l]+O(1))
            ~=~\sum_{i=1}^{l}(\ex[Y_{i}~|~U=l]+O(1)).
        \end{eqnarray}

        Let us fix an arbitrary iteration $i$, $1 \le i \le l$. We shall bound $\ex[Y_{i}~|~U=l]$ in two steps:
        First, we fix an arbitrary coloring $\chi \in PCL(H,l+1-i)$ and bound
        $\ex[Y_{i}~|~\chi_{i}=\chi ~\cap~ U = l] = 
        \ex[Y_{i}~|~\chi_{i}=\chi]$;
        second, we bound $\ex[Y_{i}~|~U=l]$.
        
        Given the coloring $\chi$, for every path $P$ in $MP(H,\chi)$, let $u_{0}(P)$ and $u_{|P|}(P)$ be the two endpoints of P, and let $c_{0}(P)$ and $c_{|P|}(P)$ be the missing colors on $u_{0}(P)$ and $u_{|P|}(P)$ from the two colors of $P$, respectively.
        We have
        \begin{eqnarray}
            \label{yconchi}
            \nonumber
            \ex[Y_{i}~|~\chi_{i}=\chi]
            &=&\sum_{P\in MP(H,\chi)}(0\cdot \pr(\mbox{$P_{i}=P$ $\cap$ $X_{i}$=0}~|~\chi_{i}=\chi) + |I(P
            )|\cdot \pr(\mbox{$P_{i}=P$ $\cap$ $X_{i}$=1}~|~\chi_{i}=\chi))\\ 
            &=&\sum_{P\in MP(H,\chi)}|I(P
            )|\cdot \pr(\mbox{$P_{i}=P$ $\cap$ $X_{i}$=1}~|~\chi_{i}=\chi).
        \end{eqnarray}

    Now, for every vertex $v$ let $low(\chi,v)$ be the number of edges of $H$ uncolored by $\chi$ with low weight, incident on $v$.
    
    Observe that for every path $P \in MP(H,\chi)$ :

    \begin{eqnarray}
        \label{pconchi}
        \nonumber
        \pr(\mbox{$P_{i}=P$ $\cap$ $X_{i}$=1}~|~\chi_{i}=\chi) &\leq~& \pr(\mbox{$c_{i}=c_{0}(P)$ $|$ $u_{i}=u_{0}(P)$ $\cap$ $X_{i}=1$ $\cap$ $\chi_{i}=\chi$})\\
        \nonumber && \cdot \pr(u_{i}=u_{0}(P)\mbox{ $\cap$ $X_{i}=1$}~|~\chi_{i}=\chi)\\
        \nonumber && + \pr(\mbox{$c_{i}=c_{|P|}(P)$ $|$ $u_{i}=u_{|P|}(P)$ $\cap$ $X_{i}=1$ $\cap$ $\chi_{i}=\chi$})\\
        \nonumber && \cdot \pr(u_{i}=u_{|P|}(P)\mbox{ $\cap$ $X_{i}$=1}~|~\chi_{i}=\chi)\\
        \nonumber&\leq& \frac{1}{\Delta_{H}/2+1}\cdot\frac{low(\chi,u_{0}(P))}{U(\chi)}+
        \frac{1}{\Delta_{H}/2+1}\cdot\frac{low(\chi,u_{|P|}(P))}{U(\chi)}
        \\
         \nonumber&\leq& \frac{1}{\Delta_{H}/2+1}\cdot\frac{3}{l+1-i}+
        \frac{1}{\Delta_{H}/2+1}\cdot\frac{3}{l+1-i}\\
        &=& O\left(\frac{1}{(l+1-i) \cdot \Delta_{H}}\right),
    \end{eqnarray}
    where the second inequality holds as the degrees of low degree vertices are smaller than $\Delta_{H}/2$ by definition, hence any low degree vertex has at least $\Delta_{H}+1-\Delta_{H}/2=\Delta_{H}/2+1$ missing colors to randomly choose from;
    the third inequality holds as $U(\chi) = l+1 -i$ by definition and as the uncolored edges of any possible $\chi_{i}$ were previously colored by (at most) three colors, hence by the validity of the coloring every vertex may have at most three uncolored edges incident on it.  

    By plugging Eq. \ref{pconchi} into Eq. \ref{yconchi}, we obtain 

    \begin{eqnarray*}
        \ex[Y_{i}~|~\chi_{i}=\chi]&=& \sum_{P\in MP(H,\chi)}|I(P)|\cdot \pr(\mbox{$P_{i}=P$ $\cap$ $X_{i}$=1}~|~\chi_{i}=\chi)\\ 
        &=& O\left(\frac{1}{(l+1-i)\cdot \Delta_{H}}\sum_{P\in MP(H,\chi)}|I(P)|\right)\\ 
        &\myeq&~~~~~ O\left(\frac{1}{(l+1-i) \cdot \Delta_{H}}\sum_{e\in \ecl(H,\chi)}w(e)\right)\\ 
        &=& O\left(\frac{1}{(l+1-i) \cdot \Delta_{H}}\sum_{e\in E_{H}}w(e)\right)\\
        &=& O\left(\frac{W_{H}}{(l+1-i) \cdot \Delta_{H}}\right).
    \end{eqnarray*}
   It follows that
    \begin{eqnarray}
    \label{eq:eqyi}
        \nonumber
        \ex[Y_{i}~|~U=l]&=&
        \sum_{\chi \in PCL(H,l+1-i)}\ex[Y_{i}~|~\chi_{i}=\chi ~\cap~ U=l]\cdot \pr(\chi_{i}=\chi~|~U=l)\\
        \nonumber
        &=& \sum_{\chi \in PCL(H,l+1-i)}\ex[Y_{i}~|~\chi_{i}=\chi]\cdot \pr(\chi_{i}=\chi~|~U=l)
        \\
        \nonumber&=&\sum_{\chi \in PCL(H,l+1-i)}O\left(\frac{W_{H}}{(l+1-i) \cdot \Delta_{H}}\right)\cdot \pr(\chi_{i}=\chi~|~U=l)\\
        \nonumber&=&O\left(\frac{W_{H}}{(l+1-i) \cdot \Delta_{H}}\right)\cdot\sum_{\chi \in PCL(H,l+1-i)}\pr(\chi_{i}=\chi~|~U=l)\\
        &=&O\left(\frac{W_{H}}{(l+1-i) \cdot \Delta_{H}}\right).
    \end{eqnarray}

Note that the number $l$ of uncolored edges is bounded by the total weight of the uncolored edges, which is  $O\left(\frac{W_{H}}{\Delta_{H}}\right)$ by Eq.\ \ref{eq:uncoloredwt}.
Now, plugging Eq.\ \ref{eq:eqyi} into Eq.\ \ref{eq:sh} yields
    
    \begin{eqnarray*}
        \ex[S_{H}~|~U=l]&=&\sum_{i=1}^{l}(\ex[Y_{i}~|~U=l]+O(1)) ~=~
        O\left(\sum_{i=1}^{l}\left(\frac{W_{H}}{(l+1-i)\Delta_{H}}+1\right) \right) \\
        &=& O\left(\frac{W_{H}}{\Delta_{H}}\sum_{i=1}^{l}\frac{1}{l+1-i}\right)+O(l) ~=~
        O\left(\frac{W_{H}}{\Delta_{H}} \cdot \log n\right)+O\left(\frac{W_{H}}{\Delta_{H}}\right)\\
        &=&O\left(\frac{W_{H}}{\Delta_{H}}\cdot\left(\frac{n}{\Delta_{H}}+\log n\right)\right).
        \end{eqnarray*}

This completes the proof of Claim \ref{low paths bound}.



    

\subsection{Concluding the Proof of Lemma~\ref{lem:timesubgraph}}

Using Eq. \ref{total repair} and Claims \ref{low fans bound} and \ref{low paths bound} we conclude that
        \begin{eqnarray*}
            \ex[T(Repair\,H)] &=&
            \sum_{l=0}^{m_{H}}\ex[T(Repair\,H)~|~U=l]\cdot \pr(U=l)\\
            &=&\sum_{l=0}^{m_{H}}\left(\sum_{i=1}^{l} O(\ex[|P_{i}|~|~U=l]) + \sum_{i=1}^{l} O(\ex[\de(u_{i})~|~U=l])\right)\cdot \pr(U=l)\\
            &=&\sum_{l=0}^{m_{H}}O\left(\frac{W_{H}}{\Delta_{H}}\cdot\left(\frac{n}{\Delta_{H}}+\log n\right)+\frac{W_{H}}{\Delta_{H}}\right)\cdot \pr(U=l)\\
            &=& \sum_{l=0}^{m_{H}} O\left(\frac{W_{H}}{\Delta_{H}}\cdot\left(\frac{n}{\Delta_{H}}+\log n\right)\right)\cdot \pr(U=l)\\
            &=& O\left(\frac{W_{H}}{\Delta_{H}}\cdot\left(\frac{n}{\Delta_{H}}+\log n\right)\right) \cdot \sum_{l=0}^{m_{H}} \pr(U=l) ~=~O\left(\frac{W_{H}}{\Delta_{H}}\cdot\left(\frac{n}{\Delta_{H}}+\log n\right)\right).
        \end{eqnarray*}

\subsection{Proof of Lemma~\ref{lem:last}}

We begin with the following corollary of the preceding claims and lemmas.

\begin{corollary} \label{cor:extimeLi}
    For any level $L_{i}$, the expected total time spent on all the subgraphs in $L_{i}$ due to calls to Algorithm \texttt{Recursive-Color-Edges} while running \texttt{Recursive-Color-Edges} on $G$ is bounded by
    $$O\left(m+\frac{W_{i}+m}{\Delta_{i}}\cdot\left(\frac{n}{\Delta_{i}}+\log n\right)\right).$$
\end{corollary}

\begin{proof}
    By  Claim \ref{cl:DeltaH},
    Claim \ref{cl:boundwt}
    and Lemma \ref{lem:timesubgraph},
    the expected total time spent on all subgraphs in $L_{i}$ due to calls to Algorithm \texttt{Recursive-Color-Edges} is bounded by

    \begin{eqnarray*}
    \sum_{H\in L_{i}}O\left(m_{H}+\frac{W_{H}}{\Delta_{H}}\cdot \left(\frac{n}{\Delta_{H}}+\log n\right)\right) &=&
    O\left(\sum_{H\in L_{i}}\left(m_{H}+\frac{W_{H}}{\Delta_{i}}\cdot \left(\frac{n}{\Delta_{i}}+\log n\right)\right)\right)\\
    &=& O\left(\sum_{H\in L_{i}}m_{H}\right)+O\left(\frac{1}{\Delta_{i}}\cdot \left(\frac{n}{\Delta_{i}}+\log n\right)\sum_{H\in L_{i}}W_{H}\right) \\&=&
    O\left(m+\frac{W_{i}+m}{\Delta_{i}}\cdot \left(\frac{n}{\Delta_{i}}+\log n\right)\right).
    \end{eqnarray*}
    
\end{proof}

\noindent
If $\Delta \leq 2\sqrt{\frac{n}{\log n}}$, the runtime is dominated by that of Algorithm \texttt{Color-Edges}, which by
    Lemma \ref{Color-Edges Bound} is
    $$O(W\log n) ~=~
    O\left(W\cdot \min\left\{\log n,\frac{\sqrt{n\log n}}{\Delta}\right\}+m\log n\right).$$

    In the complementary case $\Delta > 2\sqrt{\frac{n}{\log n}}$,
    the recursion continues until the maximum degree is at most $2\sqrt{\frac{n}{\log n}}$.
    By Claim \ref{cl:DeltaH},
    the maximum degree of every subgraph of $L_i$, for any $i$, lies in the range $[2^{-i}\Delta-2,2^{-i}\Delta+ 2]$,
    hence the number $R$ of recursion levels satisfies $R=\log {\frac{\Delta \sqrt{\log n}}{\sqrt{n}}} \pm O(1)$.

    Thus, the expected runtime of Algorithm \texttt{Recursive-Color-Edges} is given by 
\begin{eqnarray*}
    &&\sum_{i=0}^{R-1} \mbox{ expected total time spent on all subgraphs in $L_{i}$ due to calls to \texttt{Recursive-Color-Edges}}\\
    &=&\sum_{i=0}^{R-1}O\left(m+\frac{W_{i}+m}{\Delta_{i}}\cdot\left(\frac{n}{\Delta_{i}}+\log n\right)\right) \mbox{~~~(by Corollary \ref{cor:extimeLi})}\\
    &=&\sum_{i=0}^{R-1}O(m)+\sum_{i=0}^{R-1}O\left(\frac{W_{i}}{\Delta_{i}}\cdot \frac{n}{\Delta_{i}}\right)+\sum_{i=0}^{R-1}O\left(\frac{W_{i}}{\Delta_{i}}\cdot \log n\right)+\sum_{i=0}^{R-1}O\left(\frac{m}{\Delta_{i}}\cdot \frac{n}{\Delta_{i}}\right)+\sum_{i=0}^{R-1}O\left(\frac{m}{\Delta_{i}}\cdot \log n\right) \\
    &=&O(Rm)+\sum_{i=0}^{R-1}O\left(\frac{W}{\Delta}\cdot \frac{n}{2^{-i}\Delta}\right)+\sum_{i=0}^{R-1}O\left(\frac{W}{\Delta}\cdot \log n\right)+\sum_{i=0}^{R-1}O\left(\frac{m}{2^{-i}\Delta}\cdot \frac{n}{2^{-i}\Delta}\right)+\sum_{i=0}^{R-1}O\left(\frac{m}{2^{-i}\Delta}\cdot \log n\right)\\
    &=&O(Rm)+O
    \left(\frac{W \cdot n}{\Delta^{2}}\sum_{i=0}^{R-1}2^{i}\right)+O\left(\frac{W\cdot R\cdot \log n}{\Delta}\right)+O\left(\frac{mn}{\Delta^{2}}\sum_{i=0}^{R-1}4^{i}\right)+O\left(\frac{m\cdot \log n}{\Delta}\sum_{i=0}^{R-1}2^{i}\right)\\
    &=&O(Rm)+O\left(\frac{W \cdot n}{\Delta^{2}} \cdot 2^{R}\right)+O\left(\frac{W\cdot R\cdot \log n}{\Delta}\right)+O\left(\frac{mn}{\Delta^{2}} \cdot 4^{R}\right)+O\left(\frac{m\cdot \log n}{\Delta}\cdot 2^{R}\right)\\
    &=&O(Rm)+O\left(\frac{W \cdot n}{\Delta^{2}} \cdot \frac{\Delta \sqrt{\log n}}{\sqrt{n}}\right)+O\left(\frac{W\cdot R\cdot \log n}{\Delta}\right)+O\left(\frac{mn}{\Delta^{2}} \cdot \frac{\Delta^{2} \log n}{n}\right)+O\left(\frac{m\cdot \log n}{\Delta} \cdot \frac{\Delta \sqrt{\log n}}{\sqrt{n}}\right)\\
    &=&O(m\log n)+O\left(\frac{W\sqrt{n\log n}}{\Delta}\right)+O\left(\frac{W\cdot \log^{2} n}{\Delta}\right)+O(m\log n)+O(m\log n)
 \mbox{~~~(as $R = O(\log n)$)}\\
    &=&O(m\log n)+O\left(\frac{W\sqrt{n\log n}}{\Delta}\right) ~=~ O\left(W\cdot \min\left\{\log n,\frac{\sqrt{n\log n}}{\Delta}\right\}+m\log n\right),
\end{eqnarray*}
where the last inequality holds as $\Delta > 2\sqrt{\frac{n}{\log n}}$.

\end{document}